\documentclass[journal,onecolumn]{IEEEtran}

\usepackage{amsmath}
\usepackage{amssymb}
\usepackage{amsthm}

\newcommand{\ra}[1]{\renewcommand{\arraystretch}{#1}}

\newtheorem{theorem}{Theorem}

\newtheorem{lemma}[theorem]{Lemma}

\newtheorem{example}[theorem]{Example}

\newtheorem{remark}[theorem]{Remark}

\newcommand{\al}{\alpha}
\newcommand{\be}{\beta}
\newcommand{\de}{\delta}

\newcommand{\ga}{\gamma}
\newcommand{\la}{\lambda}

\newcommand{\F}{\mathbb{F}}
\newcommand{\Ftn}{\mathbb{F}_{2^n}}
\newcommand{\Ftm}{\mathbb{F}_{2^m}}

\newcommand{\Fq}{\mathbb{F}_q}

\newcommand{\Fpm}{\mathbb{F}_{p^m}}

\newcommand{\Fp}{\mathbb{F}_p}
\newcommand{\cC}{\mathcal{C}}
\newcommand{\code}{\cC_{q,d_1,d_2}}
\newcommand{\co}{\code^{\perp}}

\newcommand{\Tr}{\text{Tr}}
\newcommand{\Tm}{\Tr_m}
\newcommand{\Tn}{\Tr_n}

\newcommand{\zp}{\zeta_p}

\newcommand{\Nt}{N_2(q,d_1,d_2)}
\newcommand{\Nth}{N_3(q,d_1,d_2)}
\newcommand{\x}{\bar{x}}
\newcommand{\ba}{\bar}

\begin{document}

\title{On the Weight Distribution of Cyclic Codes with Niho Exponents}

\author{Shuxing Li, Tao Feng and Gennian Ge
\thanks{The research of T. Feng was supported by Fundamental Research Fund for the Central Universities of China,
Zhejiang Provincial Natural Science Foundation under Grant LQ12A01019, the National Natural Science Foundation
of China under Grant 11201418, and the Research Fund for Doctoral Programs from the Ministry of Education of
China under Grant 20120101120089. The research of G. Ge was supported by the National Natural Science Foundation of China under Grant No.~61171198 and Zhejiang Provincial Natural Science Foundation of China under Grant No.~LZ13A010001.}
\thanks{S. Li is with the Department of Mathematics, Zhejiang University,
Hangzhou 310027,  China (e-mail: sxli@zju.edu.cn).}
\thanks{T. Feng is with the Department of Mathematics, Zhejiang University,
Hangzhou 310027,  China (e-mail: tfeng@zju.edu.cn). He is also with Beijing Center for Mathematics and Information Interdisciplinary Sciences, Beijing, 100048, China.}
\thanks{G. Ge is with  the School of Mathematical Sciences, Capital Normal University,
Beijing, 100048, China (e-mail: gnge@zju.edu.cn). He is also with Beijing Center for Mathematics and Information Interdisciplinary Sciences, Beijing, 100048, China.}
}

\maketitle

\begin{abstract}
Recently, there has been intensive research on the weight distributions of cyclic codes. In this paper, we compute the weight distributions of three classes of cyclic codes with Niho exponents. More specifically, we obtain two classes of binary three-weight and four-weight cyclic codes and a class of nonbinary four-weight cyclic codes. The weight distributions follow from the determination of value distributions of certain exponential sums. Several examples are presented to show that some of our codes are optimal and some have the best known parameters.
\end{abstract}


\section{Introduction}\label{sec1}

Cyclic codes are a special class of linear codes with preferable algebraic properties. In favor of practical use, cyclic codes enjoy efficient encoding and decoding algorithms. They have been widely used in many areas such as communication and data storage system. Moreover, cyclic codes are employed to construct other interesting structures, such as quantum codes \cite{TM}, frequency hopping sequences \cite{DYT} and so on.

For a cyclic code $\cC$ of length $l$ over some finite field $\Fp$, each codeword $c=(c_0,\ldots,c_{l-1})$ can be identified with a polynomial $\sum_{i=0}^{l-1} c_ix^i \in \Fq[x]$. Indeed, $\cC$ is an ideal of the principle ideal domain $\Fp[x]/(x^{l}-1)$. Thus, it can be expressed as $\cC=(g(x))$, where $g(x) \in \Fp[x]$ with $g(x) \mid x^l-1$ is called the {\em generator polynomial} of $\cC$. A cyclic code $\cC$ is said to have $i$ {\em zeros} if its generator polynomial can be factorized as a product of $i$ irreducible polynomials over $\Fp$. When its dual code $\cC^{\perp}$ has $i$ zeros, we call $\cC$ as a cyclic code with $i$ {\em nonzeros}. A cyclic code $\cC$ is irreducible if it has one nonzero and reducible otherwise.

Let $A_i$ be the number of codewords in $\cC$ with Hamming weight $i$, where $0 \le i \le l$. The weight distribution $\{A_0,A_1,\ldots,A_l\}$ is an important research subject in coding theory. For irreducible cyclic codes, it is pointed out by McEliece \cite{Mc} that their weights can be expressed via Gauss sums. While there are many results concerning the weight distributions of irreducible cyclic codes, we refer the readers to a comprehensive survey $\cite{DY}$ and the references therein.

For reducible cyclic codes with few nonzeros, their weight distributions have been intensively studied, including \cite{DGZ,DLMZ,FLX,FM,HX,LZH,LF1,LF2,LTW,MZLFD,M,Ve,VM,WTQYX,Wo,X,ZHJYC,ZLH,ZDLZ}. Basically, the weight distribution is closely related to the value distribution of certain exponential sum, which is difficult to compute in general. Thus, the study of weight distributions stimulates the development of delicate techniques concerning the computation of exponential sums in recent years. For instance, Luo and Feng \cite{LF1,LF2} proposed an elegant method employing quadratic forms to compute the value distribution. Their idea inspires a series of works following this line \cite{DGZ,LTW,ZHJYC,ZLH,ZDLZ}. In \cite{DLMZ,MZLFD}, the authors express the weights of cyclic codes via Gauss period. This observation leads to further studies in \cite{FM,Ve,VM,WTQYX,X}. In a word, motivated by these original ideas, much progress has been made recently.

In this paper, we consider the weight distribution of certain cyclic codes with two nonzeros. We fix $n=2m$, where $m$ is a positive integer. Let $p$ be a prime and $q=p^n$ be a prime power.  We use $\Fq$ to denote the finite field of order $q$ and fix $\theta$ to be a primitive element of $\Fq$. We use $\code$ to denote the cyclic code of length $q-1$ with two zeros $\theta^{d_1}$ and $\theta^{d_2}$. Namely, the generator polynomial of $\code$ is $g_{d_1}(x)g_{d_2}(x)$, where $g_i(x)$ is the minimal polynomial of $\theta^i$ over $\Fp$. By the Pless power moment identities \cite{P}, determining the weight distribution of $\code$ is equivalent to determining that of its dual code $\code^{\perp}$, which is a reducible cyclic code with two nonzeros. Usually, it is convenient to study the dual code $\code^{\perp}$, since it owns a simple trace representation due to Delsarte \cite{Del}.

Given a prime $p$, a positive integer $d$ is of {\it Niho-type} if $d \equiv p^i \pmod {p^m-1}$ for some integer $i$. Without loss of generality, we can assume that $d \equiv 1 \pmod {p^m-1}$. For two Niho exponents $d=s(p^m-1)+1$ and $d^{\prime}=s^{\prime}(p^m-1)+1$, we call them {\em equivalent} if $d^{\prime} \equiv p^id \pmod{p^n-1}$ for some integer $i$. Moreover, $d^{\prime} \equiv p^md \pmod{p^n-1}$ if and only if $s+s^{\prime} \equiv 1 \pmod {p^m+1}$. Hence, we can restrict $s$ in the range $1 \le s \le p^{m-1}+1$. For a Niho exponent $d=s(p^m-1)+1$ with $(d,p^n-1)=1$, its inverse $d^{-1}=s^{\prime}(p^m-1)+1$ is also of Niho type, where $s^{\prime}\equiv \frac{s}{2s-1} \pmod {p^m+1}$ and $\frac{1}{2s-1}$ represents the inverse of $2s-1$ module $p^m+1$. The term Niho-type stems from the study of Niho which concerns the cross correlation distribution between a maximal length sequence ($m$-sequence) and its decimation \cite{Niho}. Let $\zp$ be the $p$-th complex root of unity. If $(d_1,q-1)=(d_2,q-1)=1$, then the weight distribution of $\code^{\perp}$ can be obtained from the value distribution of
$$
\sum_{x\in \Fq} \zp^{\Tn(ax+x^{d_1^{-1}d_2})}, \quad a \in \Fq,
$$
which identifies with the cross correlation distribution between a pair of $m$-sequences with Niho-type decimation $d_1^{-1}d_2$.

It is worthy noting that there are a few papers concerning cyclic codes with Niho exponents. In \cite{C}, Charpin considers the weight distribution of $\cC_{2^n,d_1,1}^{\perp}$ with $(d_1,2^n-1)=1$. It is proved that this code has at least four nonzero weights. In \cite{LZH}, Li et al. consider a class of binary cyclic codes with three nonzeros and Niho exponents, and they obtain the weight distribution.

This paper concerns the weight distribution of $\code^{\perp}$, where $d_1$ and $d_2$ are both of Niho-type. We observe that the Niho exponents $d_1$ and $d_2$ need not to be coprime with $q-1$. By specifying certain conditions on $d_1$ and $d_2$, we obtain the weight distributions of two classes of binary cyclic codes and a class of nonbinary cyclic codes. The weight distributions are determined by computing the value distributions of
$$
S(a,b)=\sum_{x\in \Ftn} (-1)^{\Tm(ax^{2^m+1})+\Tn(bx^{d_2})}
$$
and
$$
T(a,b)=\sum_{x\in \Fq} \zp^{\Tn(ax^{d_1}+bx^{d_2})},
$$
where $\Tm$ (resp. $\Tn$) is the absolute trace from $\F_{p^m}$ (resp. $\Fq$) to $\F_p$. Moreover, several examples are presented to show that some of our binary cyclic codes are optimal linear codes or have the best known parameters.

The rest of this paper is organized as follows. In Section \ref{sec2}, we present some preliminaries including Delsarte's Theorem, Niho's Theorem and the Pless moment identities. The general strategy for our computation of weight distributions will be outlined. In Section \ref{sec3}, we calculate the weight distributions of two classes of binary cyclic codes with Niho exponents. Several examples are provided to show that some of our codes are either optimal or having the best known parameters. In Section \ref{sec4}, we derive the weight distribution of a class of nonbinary cyclic codes with Niho exponents. Section \ref{sec5} concludes the paper.

\section{Preliminaries}\label{sec2}

This section is devoted to some preliminaries. In the first part, we fix some notations. In the second part, we introduce Delsarte's Theorem and Niho's Theorem. A generalization of Niho's Theorem over odd characteristic is also presented. Based on Delsarte's Theorem, determining weight distributions can be translated into the computation of value distributions of certain exponential sums. Meanwhile, Niho's Theorem builds an elegant connection between the values of these exponential sums and the solutions of certain equations. Thus, we can determine the values by analysing the corresponding equation. In the third part, we introduce some moment identities. These moment identities are used to compute the frequencies of these values.

\subsection{Notations}

In this subsection, we fix some notations which will be used throughout the rest of this paper. Let $m$ be a positive integer and fix $n=2m$. Let $p$ be a prime and $q=p^n$. Let $\Fq$ be the finite field of order $q$ and $\theta$ be a primitive element of $\Fq$. Define the set of squares (resp. nonsquares) in $\Fq$ as $Q$ (resp. $NQ$). When $p$ is an odd prime, for each $x \in \Fq^*$, there are exactly two elements in $\Fq^*$ whose square equal to $x$. We denote them by $\pm x^{\frac12}$.

Define $S=\{x \in \Fq | x\x=1\}$, where $\x=x^{p^m}$. Thus, $S$ is a cyclic group of order $p^m+1$. In addition, for any positive integer $l$, we set $S_l=\{x^l \mid x \in S\}$.

Given a positive integer $d$, we use $cl(d)$ to denote the least positive integer $k$ such that $2^kd \equiv d \pmod {2^n-1}$.

We use $\Tn$ (resp. $\Tm$) to denote the absolute trace from $\F_q$ (resp. $\Fpm$) to $\F_p$. Let $\zp$ denote the $p$-th complex root of unity. We consider the following two exponential sums:
$$
S(a,b)=\sum_{x\in \Ftn} (-1)^{\Tm(ax^{2^m+1})+\Tn(bx^{d_2})}
$$
and
$$
T(a,b)=\sum_{x\in \Fq} \zp^{\Tn(ax^{d_1}+bx^{d_2})}.
$$
To make it more clear, we write
$$
T_1(a,b)=\sum_{x\in \Ftn} (-1)^{\Tn(ax^{d_1}+bx^{d_2})}
$$
and
$$
T_2(a,b)=\sum_{x\in \Fq} \zp^{\Tn(ax^{d_1}+bx^{d_2})},
$$
where $p$ is an odd prime in $T_2(a,b)$.

\subsection{Delsarte's Theorem and Niho's Theorem}

For a cyclic code $\co$, there is a nice trace representation of its codewords. More precisely, by Delsarte's Theorem \cite{Del}, we have
\begin{equation*}
  \co=\{c(a,b)=(\Tn(a\theta^{id_1}+b\theta^{id_2}))_{i=0}^{q-2} \mid a,b \in \Fq \}.
\end{equation*}
 The Hamming weight of a codeword $c(a,b)$ can be expressed as
\begin{align*}
w_H(c(a,b)) &= (q-1)-\frac{1}{p}\sum_{x\in \Fq^*}\sum_{\la \in \Fp} \zp^{\la \Tn(a x^{d_1}+b x^{d_2})}\\
            &= (q-1)(1-\frac{1}{p})-\frac{1}{p}\sum_{\la \in \Fp^*}\sum_{x\in \Fq^*} \zp^{\Tn(\la a x^{d_1}+\la b x^{d_2})}.
\end{align*}
Consequently, the information of the weight distribution can be obtained from the value distribution of
$$
\sum_{x\in \Fq^*} \zp^{\Tn(a x^{d_1}+b x^{d_2})},\quad a,b \in \Fq.
$$
Next, we will see that when $d_1$ and $d_2$ are Niho exponents, the possible values of this exponential sum are determined by the solutions of certain equation.

At first, we consider the case $p=2$.  The {\em polar representation} says that each $x \in \Ftn^*$ can be uniquely represented as $x=yz$, where $y\in \Ftm^*$ and $z \in S$. This fact is a key ingredient of the following lemma which is  essentially proposed by Niho \cite{Niho}. Here we provide a short proof to make this paper self-contained.

\begin{lemma}\label{Niho}
Let $p=2$ and $q=2^n$.
\begin{enumerate}
\item[1)] If $d_2=s_2(2^m-1)+1$, we have $S(a,b)=(U(a,b)-1)2^m$, where $U(a,b)$ is the number of $z \in S$ satisfying
$$
\ba{b}z^{2(2s_2-1)}+a^{\frac12}z^{2s_2-1}+b=0.
$$
\item[2)] If $d_1=s_1(2^m-1)+1$ and $d_2=s_2(2^m-1)+1$, we have $T_1(a,b)=(V(a,b)-1)2^m$, where $V(a,b)$ is the number of $z \in S$ satisfying
$$
\ba{b}z^{2s_2-1}+\ba{a}z^{s_1+s_2-1}+az^{s_2-s_1}+b=0.
$$
\end{enumerate}
\end{lemma}
\begin{proof}
We only prove 2) since the proof of 1) is analogous. For each $x \in \Ftn^*$, we can write $x=yz$, where $y \in \Ftm^*$ and $z \in S$. Therefore,
\begin{align*}
T_1(a,b) &= 1+\sum_{x \in \Ftn^*} (-1)^{\Tn(ax^{d_1}+bx^{d_2})}\\
       &= 1+\sum_{y \in \Ftm^*}\sum_{z \in S} (-1)^{\Tn(ayz^{d_1}+byz^{d_2})}\\
       &= 1-|S|+\sum_{z \in S}\sum_{y \in \Ftm} (-1)^{\Tn((az^{1-2s_1}+bz^{1-2s_2})y)}\\
       &= -2^m+\sum_{z \in S}\sum_{y \in \Ftm} (-1)^{\Tm((az^{1-2s_1}+bz^{1-2s_2}+\ba{a}z^{2s_1-1}+\ba{b}z^{2s_2-1})y)}\\
       &= -2^m+ |\{z \in S \mid az^{1-2s_1}+bz^{1-2s_2}+\ba{a}z^{2s_1-1}+\ba{b}z^{2s_2-1}=0 \}|\cdot 2^m\\
       &=(V(a,b)-1)2^m.
\end{align*}
\end{proof}

Secondly, we consider the case where $p$ is an odd prime. The situation is slightly different since the polar representation does not hold when $p$ is odd. Instead, each $x \in Q$ (resp. $x \in NQ$) can be expressed twice as $x=yz$ or $x=(-y)(-z)$ (resp. $x=\theta yz$ or $x=\theta (-y)(-z)$), where $y$ ranges over $\Fpm^*$ and $z$ ranges over $S$. Therefore, we have
$$
2*\Fq^*=\{yz|y \in \Fpm^*, z \in S\} \cup \{\theta yz|y \in \Fpm^*, z \in S\},
$$
where $2*\Fq^*$ is the multiset in which each element of $\Fq^*$ appears twice and the two sets on the right hand side are regarded as multisets. A modification of Lemma~\ref{Niho} leads to the following lemma.

\begin{lemma}\label{Nihopary}
Let $p$ be an odd prime and $q=p^n$. Suppose $d_1=s_1(p^m-1)+1$ and $d_2=s_2(p^m-1)+1$. Then for any $\la \in \Fp^*$, we have $T_2(\la a,\la b)=(W(a,b)-1)p^m$, where $W(a,b)$ is the number of $u \in S$ satisfying
$$
\ba{b}u^{2s_2-1}+\ba{a}u^{s_1+s_2-1}+au^{s_2-s_1}+b=0.
$$
\end{lemma}
\begin{proof}
For any $\la \in \Fp^*$, we have
\begin{align*}
  T_2(\la a, \la b) &= 1+\sum_{x \in \Fq^*} \zp^{\Tn(\la ax^{d_1}+\la bx^{d_2})}\\
                  &= 1+\frac12\sum_{y \in \Fpm^*}\sum_{z \in S} (\zp^{\Tn(\la (ayz^{d_1}+byz^{d_2}))}+\zp^{\Tn(\la( a\theta^{d_1} yz^{d_1}+b\theta^{d_2}yz^{d_2}))})\\
                  &= 1-|S|+\frac12\sum_{z \in S}\sum_{y \in \Fpm} (\zp^{\Tn((az^{1-2s_1}+bz^{1-2s_2})\la y)}+\zp^{\Tn((a\theta^{d_1}z^{1-2s_1}+b\theta^{d_2}z^{1-2s_2})\la y)})\\
                  &= -p^m+\frac12\sum_{z \in S}\sum_{y \in \Fpm} (\zp^{\Tm((az^{1-2s_1}+bz^{1-2s_2}+\ba{a}z^{2s_1-1}+\ba{b}z^{2s_2-1})\la y)}\\
                  &\quad+ \zp^{\Tm((a\theta^{d_1}z^{1-2s_1}+b\theta^{d_2}z^{1-2s_2}+\ba{a}\ba{\theta}^{d_1}z^{2s_1-1}+\ba{b}\ba{\theta}^{d_2}z^{2s_2-1})\la y)}).
\end{align*}
Denote the number of $z \in S$ satisfying
$$
az^{1-2s_1}+bz^{1-2s_2}+\ba{a}z^{2s_1-1}+\ba{b}z^{2s_2-1}=0
$$
by $W_1(a,b)$ and the number of $z \in S$ satisfying
$$
a\theta^{d_1}z^{1-2s_1}+b\theta^{d_2}z^{1-2s_2}+\ba{a}\ba{\theta}^{d_1}z^{2s_1-1}+\ba{b}\ba{\theta}^{d_2}z^{2s_2-1}=0
$$
by $W_2(a,b)$. We have
$$
T_2(\la a, \la b)=(\frac{W_1(a,b)+W_2(a,b)}{2}-1)p^m.
$$
Thus, it remains to prove that $W(a,b)=\frac{W_1(a,b)+W_2(a,b)}{2}$. Direct computation shows that the above two equations are respectively equivalent to
\begin{equation}\label{eqn1}
\ba{b}z^{2(2s_2-1)}+\ba{a}z^{2(s_1+s_2-1)}+az^{2(s_2-s_1)}+b=0
\end{equation}
and
\begin{equation}\label{eqn2}
\ba{b}\eta^{2s_2-1}z^{2(2s_2-1)}+\ba{a}\eta^{s_1+s_2-1}z^{2(s_1+s_2-1)}+a\eta^{s_2-s_1}z^{2(s_2-s_1)}+b=0,
\end{equation}
where $\eta=\theta^{-(p^m-1)}$ is a generator of $S$. Set $u=z^2$. Equation ($\ref{eqn1}$) becomes
\begin{equation}\label{eqn3}
\ba{b}u^{2s_2-1}+\ba{a}u^{s_1+s_2-1}+au^{s_2-s_1}+b=0,
\end{equation}
where $u \in S_{2}$. For each solution $u$ of (\ref{eqn3}), it corresponds to two solutions $\pm u^{\frac12}$ of (\ref{eqn1}). In the same way, Equation ($\ref{eqn2}$) becomes
\begin{equation}\label{eqn4}
\ba{b}(\eta u)^{2s_2-1}+\ba{a}(\eta u)^{s_1+s_2-1}+a(\eta u)^{s_2-s_1}+b=0,
\end{equation}
where $\eta u \in S \setminus S_{2}$. For each solution $\eta u$ of (\ref{eqn4}), it corresponds to two solutions $\pm u^{\frac12}$ of (\ref{eqn2}). Note that the solutions of Equation ($\ref{eqn3}$) (resp. Equation ($\ref{eqn4}$)) are exactly the solutions of
$$
\ba{b}u^{2s_2-1}+\ba{a}u^{s_1+s_2-1}+au^{s_2-s_1}+b=0
$$
belonging to $S_{2}$ (resp. $S \setminus S_{2}$). Thus, we deduce $W(a,b)=\frac{W_1(a,b)+W_2(a,b)}{2}$ and the proof is now complete.
\end{proof}

\subsection{Moment Identities}

From now on, we use $\Nt$ to denote the number of solutions to the equations
\begin{equation}
\left\{\begin{array}{c}
             x^{d_1}+y^{d_1}=0 \\
             x^{d_2}+y^{d_2}=0
             \end{array}
\right., \quad x,y \in \Fq.
\end{equation}
Similarly, let $\Nth$ denote the number of solutions to the equations
\begin{equation}
\left\{\begin{array}{c}
             x^{d_1}+y^{d_1}+z^{d_1}=0 \\
             x^{d_2}+y^{d_2}+z^{d_2}=0
             \end{array}
\right., \quad x,y,z \in \Fq.
\end{equation}

The following moment identities play an important role in the determination of weight distributions.

\begin{lemma}\label{moment}
Let $p$ be an odd prime and $q=p^n$. Then we have
\begin{enumerate}
\item[1)] $\sum_{a \in \Ftm} \sum_{b \in \Ftn} S(a,b)=2^{3m}$.
\item[2)] $\sum_{a \in \Ftm} \sum_{b \in \Ftn} S(a,b)^2=2^{3m}N_{2}(2^n,2^m+1,d_2)$.
\item[3)] $\sum_{a,b \in \Ftn} T_1(a,b)=2^{2n}$.
\item[4)] $\sum_{a,b \in \Ftn} T_1(a,b)^2 =2^{2n}N_2(2^n,d_1,d_2)$.
\item[5)] $\sum_{a,b \in \Ftn} T_1(a,b)^3=2^{2n}N_3(2^n,d_1,d_2)$.
\item[6)] $\sum_{a,b \in \Fq} T_2(a,b)=p^{2n}$.
\item[7)] $\sum_{a,b \in \Fq} T_2(a,b)^2 =p^{2n}\Nt$.
\item[8)] $\sum_{a,b \in \Fq} T_2(a,b)^3=p^{2n}\Nth$.
\end{enumerate}
\end{lemma}
\begin{proof}
The proof is routine and analogous to that of \cite[Lemma 4]{LTW}. So we omit it here.
\end{proof}

Consequently, precise information of these moment identities is available if we can count the number of solutions of certain equation systems.

\section{Binary Cyclic Codes With Niho Exponents}\label{sec3}

Considering a Niho exponent $d=s(2^m-1)+1$, it is straightforward to verify that
$$
cl(d)=\begin{cases}
              m  & \text{if $s\equiv \frac12 \pmod{2^m+1}$}, \\
              n  & \text{otherwise},
              \end{cases}
$$
where $\frac12$ represents the inverse of $2$ modulo $2^m+1$.

This section concerns the weight distributions of binary cyclic codes with Niho exponents. The first part studies the weight distribution of $\cC_{2^n,d_1,d_2}^{\perp}$ with $cl(d_1)=m$ and $cl(d_2)=n$. For this purpose, we compute the value distribution of $S(a,b)$. In the second part, we consider the weight distribution of $\cC_{2^n,d_1,d_2}^{\perp}$ with $cl(d_1)=cl(d_2)=n$. By imposing some specific conditions on $d_1$ and $d_2$, we obtain the value distribution of $T_1(a,b)$. Thus, the weight distribution of related cyclic codes follows immediately.

\subsection{The Value Distribution of $S(a,b)$ and Related Cyclic Codes}

Throughout this subsection, we consider the value distribution of $S(a,b)$ with $d_2=s_2(2^m-1)+1$. To ensure that $2^m+1$ and $d_2$ are not equivalent, we have $s_2 \not\equiv \frac12 \pmod {2^m+1}$. As a preparation, we have the following lemma.

\begin{lemma}\label{N2}
Suppose $q=2^n$ and $l=(2s_2-1,2^m+1)$. Then $N_2(q,2^m+1,d_2)=(2^n-1)l+1$.
\end{lemma}
\begin{proof}
By definition, $N_2(q,2^m+1,d_2)$ is the number of solutions to the equations
\begin{equation}\label{N2eqn1}
\left\{\begin{array}{c}
             x^{2^m+1}+y^{2^m+1}=0 \\
             x^{d_2}+y^{d_2}=0
             \end{array}
             \right., \quad x,y \in \Fq.
\end{equation}
When $y=0$, we have one solution $(x,y)=(0,0)$. When $y \in \Ftn^*$, by setting $z=\frac{x}{y}$, we only need to consider the system
\begin{equation}\label{N2eqn2}
\left\{\begin{array}{c}
             z^{2^m+1}=1 \\
             z^{d_2}=1
             \end{array}
             \right., \quad z \in \Fq.
\end{equation}
Each solution of (\ref{N2eqn2}) corresponds to $2^n-1$ solutions of (\ref{N2eqn1}). Since $l=(2s_2-1,2^m+1)=(d_2,2^m+1)$, (\ref{N2eqn2}) is equivalent to $z^{l}=1$, which has exactly $l$ solutions in $\Fq$. Hence, we deduce that $N_2(q,2^m+1,d_2)=(2^n-1)l+1$.
\end{proof}

We are now ready to determine the value distribution of
$$
S(a,b)=\sum_{x\in \Ftn} (-1)^{\Tm(ax^{2^m+1})+\Tn(bx^{d_2})}.
$$

\begin{theorem}\label{vdbinarythree}
Assume $n=2m$ with $m \ge 1$. Define $d_2=s_2(2^m-1)+1$ with $s_2 \not\equiv \frac12 \pmod {2^m+1}$. Set $q=2^n$ and $l=(2s_2-1,2^m+1)$. Then the value distribution of $S(a,b)$ is listed in Table \ref{table1}.
\end{theorem}
\begin{proof}
By 1) of Lemma~\ref{Niho}, we have $S(a,b)=(U(a,b)-1)2^m$, where $U(a,b)$ is the number of $z \in S$ satisfying
$$
\ba{b}z^{2(2s_2-1)}+a^{\frac12}z^{2s_2-1}+b=0.
$$
When $(a,b)=(0,0)$, it is easy to see that $U(a,b)=2^m+1$ and $S(a,b)$ takes the trivial value $2^{2m}$. Below, we consider the case $(a,b)\ne (0,0)$. Setting $u=z^{2s_2-1}$, the equation becomes
$$
\ba{b}u^{2}+a^{\frac12}u+b=0,
$$
which has either $0,1$ or $2$ solutions in $S_l$. Since $l=(2s_2-1,2^m+1)$, for any $u \in S_l$, the equation $z^{2s_2-1}=u$ has exactly $l$ solutions in $S$. Therefore, we have $U(a,b)\in \{0,l,2l\}$ when $(a,b)\ne (0,0)$. Consequently, $S(a,b)$ takes three distinct values $\{-2^m,(l-1)2^m,(2l-1)2^m\}$ when $(a,b)\ne (0,0)$. The corresponding frequencies of these values can be obtained from Lemma~\ref{moment} and Lemma~\ref{N2}. The proof is now complete and we list the value distribution in Table \ref{table1}.
\end{proof}

\begin{table*}
\begin{center}
\ra{1.8}
\caption{Value distribution of Theorem~\ref{vdbinarythree}}\label{table1}
\begin{tabular}{|c|c|}
\hline
Value  &  Frequency \\ \hline
$2^{2m}$     &  $1$ \\ \hline
$(2l-1)2^{m}$ & $\frac{(2^{2m}-1)(2^m-l+1)}{2l^2}$ \\ \hline
$(l-1)2^{m}$ & $\frac{(2^{2m}-1)((2^m+2)l-2^m-1)}{l^2}$ \\ \hline
$-2^{m}$ & $2^{3m}-1+\frac{(2^{2m}-1)(2^{m}+1-(2^{m+1}+3)l)}{2l^2}$ \\ \hline
\end{tabular}
\end{center}
\end{table*}

As a direct consequence of Theorem~\ref{vdbinarythree}, we obtain the weight distribution of a class of binary cyclic codes.

\begin{theorem}\label{thmbinarythree}
Assume $n=2m$ with $m \ge 1$. Define $d_1=2^m+1$ and $d_2=s_2(2^m-1)+1$ with $s_2 \not\equiv \frac12 \pmod {2^m+1}$. Set $q=2^n$ and $l=(2s_2-1,2^m+1)$. Then $\code^{\perp}$ is a $[2^n-1,3m,2^{2m-1}-(2l-1)2^{m-1}]$ binary code. Its weight distribution is listed in Table \ref{table2}.
\end{theorem}

\begin{table*}
\begin{center}
\ra{1.8}
\caption{Weight distribution of Theorem~\ref{thmbinarythree}}\label{table2}
\begin{tabular}{|c|c|}
\hline
Weight  &  Frequency \\ \hline
$0$     &  $1$ \\ \hline
$2^{2m-1}-(2l-1)2^{m-1}$ & $\frac{(2^{2m}-1)(2^m-l+1)}{2l^2}$ \\ \hline
$2^{2m-1}-(l-1)2^{m-1}$ & $\frac{(2^{2m}-1)((2^m+2)l-2^m-1)}{l^2}$ \\ \hline
$2^{2m-1}+2^{m-1}$ & $2^{3m}-1+\frac{(2^{2m}-1)(2^{m}+1-(2^{m+1}+3)l)}{2l^2}$ \\ \hline
\end{tabular}
\end{center}
\end{table*}


Given $m$, the code is determined by one parameter $s_2$. From now on, we refer the code table as the one maintained by Grassl \cite{Gra}. We present some examples concerning the weight distributions of the cyclic codes derived from the above theorem. According to the code table, some of them are optimal linear codes.

\begin{example}
When $m=2$, we have $s_2 \in \{1,2\}$. Then $l=(2s_2-1,2^m+1)=1$ for both choices of $s_2$. The corresponding two cyclic codes are $[15,6,6]$ binary codes with the same weight distribution:
$$
1+30x^6+15x^8+18x^{10}.
$$
Referring to the code table \cite{Gra}, our cyclic codes are optimal.
\end{example}

\begin{example}
When $m=3$, we have $s_2 \in \{1,2,3,4\}$. Furthermore, we have $l=(2s_2-1,2^m+1)=1$ for $s_2 \in \{1,3,4\}$. The corresponding three cyclic codes are $[63,9,28]$ binary codes with the same weight distribution:
$$
1+252x^{28}+63x^{32}+196x^{36}.
$$
Referring to the code table \cite{Gra}, our cyclic codes are optimal.
\end{example}

\subsection{The Value Distribution of $T_1(a,b)$ and Related Cyclic Codes}

Now, we compute the value distribution of $T_1(a,b)$ in one special case. Throughout this subsection, we fix $d_1=s_1(2^m-1)+1$ and $d_2=s_2(2^m-1)+1$ where $s_1=2^{k-1}t-\frac{t-1}{2}$ and $s_2=2^{k-1}t+\frac{t+1}{2}$ for some positive integer $k$ and some odd number $t \ge 1$. To ensure that $d_1, d_2$ are not equivalent and $cl(d_1)=cl(d_2)=n$, we have $(2^k-1)t,(2^k+1)t \not\equiv 0 \pmod{2^m+1}$. We call two pairs of Niho exponents $(d_1,d_2)$ and $(d_1^{\prime},d_2^{\prime})$ {\em equivalent} if $(d_1,d_1^{\prime})$ and $(d_2,d_2^{\prime})$ are pairwise equivalent or $(d_1,d_2^{\prime})$ and $(d_2,d_1^{\prime})$ are pairwise equivalent. Set $s_1=2^{k-1}t-\frac{t-1}{2}$, $s_2=2^{k-1}t+\frac{t+1}{2}$, $s_1^{\prime}=2^{k+m-1}t-\frac{t-1}{2}$ and $s_2^{\prime}=2^{k+m-1}t+\frac{t+1}{2}$. It is easy to see that $s_1+s_2^{\prime} \equiv 1 \pmod {2^m+1}$ and $s_1^{\prime}+s_2 \equiv 1 \pmod {2^m+1}$. Namely, $k$ and $k+m$ produce two equivalent pairs of Niho exponents. Thus, we can restrict $k$ in the range $1 \le k \le m$. A similar analysis shows that we can assume $1\le t \le 2^m+1$ without loss of generality. Below, we will determine the value distribution of $T_1(a,b)$ with some more conditions imposed.

As a preparation, we have the following lemma.

\begin{lemma}\label{N3}
Suppose $q=2^n$ and $l=(t,2^m+1)$. Then
\begin{enumerate}
\item[1)] $\Nt=(2^n-1)l+1$.
\item[2)] $\Nth=(2^m-2)(2^{n}-1)l^2+3(2^{n}-1)l+1$.
\end{enumerate}
\end{lemma}
\begin{proof}
1) Note that
$$
(d_1,2^n-1)=((2^k-1)t,2^m+1)
$$
and
$$
(d_2,2^n-1)=((2^k+1)t,2^m+1).
$$
Thus, $l$ divides both $(d_1,2^n-1)$ and $(d_2,2^n-1)$. Moreover, we have either $(d_1,2^n-1)=l$ or $(d_2,2^n-1)=l$. Hence, the system
$$
\left\{\begin{array}{c}
             u^{d_1}=1 \\
             u^{d_2}=1
             \end{array}
\right.
$$
has exactly $l$ solutions. Following the same spirit of the proof in Lemma~\ref{N2}, the remaining part is routine. \\
2) By definition, $\Nth$ is the number of solutions to the equations
\begin{equation}\label{N3eqn1}
\left\{\begin{array}{c}
             x^{d_1}+y^{d_1}+z^{d_1}=0 \\
             x^{d_2}+y^{d_2}+z^{d_2}=0
             \end{array}
\right., \quad x,y,z \in \Fq.
\end{equation}
When $z=0$, there are $\Nt=(2^n-1)l+1$ solutions.

When $z\ne 0$, the situation is more involved. By setting $u=\frac{x}{z}$ and $v=\frac{y}{z}$, we only need to consider the system
\begin{equation}\label{sysbianry}
\left\{\begin{array}{c}
             u^{d_1}+v^{d_1}=1 \\
             u^{d_2}+v^{d_2}=1
             \end{array}
\right., \quad u,v \in \Fq.
\end{equation}
Each solution of (\ref{sysbianry}) corresponds to $2^n-1$ solutions of (\ref{N3eqn1}). If $u=0$ or $v=0$, by the proof of 1), the system (\ref{sysbianry}) has exactly $l$ solutions.
If $uv \ne 0$, by the polar representation, $u$ and $v$ can be uniquely expressed as $u=\al\de$ and $v=\be\ga$, where $\al,\be \in \Ftm^*$ and $\de,\ga \in S$. Thus the system (\ref{sysbianry}) is equivalent to
\begin{equation}\label{sysbinary2}
\left\{\begin{array}{c}
  \al\de^{-t(2^k-1)}+\be\ga^{-t(2^k-1)}=1 \\
  \al\de^{-t(2^k+1)}+\be\ga^{-t(2^k+1)}=1
\end{array}
\right..
\end{equation}
Note that
$$
\Delta=
\begin{vmatrix}
  \de^{-t(2^k-1)} & \ga^{-t(2^k-1)}\\
  \de^{-t(2^k+1)} & \ga^{-t(2^k+1)}
\end{vmatrix}
=\de^{-t(2^k-1)}\ga^{-t(2^k+1)}-\de^{-t(2^k+1)}\ga^{-t(2^k-1)}.
$$
Below, we split our discussion into two cases.

If $\Delta=0$, we have $\de^t=\ga^t$. Comparing with (\ref{sysbinary2}), we have $\de^t=\ga^t=1$ and the system (\ref{sysbinary2}) degenerates to $\al+\be=1$. Note that there are $l^2$ pairs of $(\de,\ga)$ such that $\de^t=\ga^t=1$. Moreover, for each pair $(\de,\ga)$, there are $2^m-2$ pairs of $(\al,\be)$, such that $\al+\be=1$ and $\al\be \ne 0$. Hence, there are $(2^m-2)l^2$ solutions in this case.

If $\Delta \ne 0$, i.e., $\de^t \ne \ga^t$, solving the system (\ref{sysbinary2}) yields
\begin{align*}
\al&=\frac{1+\ga^{2t}}{\de^{-t(2^k-1)}(1+\de^{-2t}\ga^{2t})},\\ \be&=\frac{1+\de^{2t}}{\ga^{-t(2^k-1)}(1+\de^{2t}\ga^{-2t})}.
\end{align*}
We are going to show that no solution exists in this case. Since $\al\in \Ftm^*$, we have $\al=\ba{\al}$, which leads to
$\de^t=1$. Similarly, since $\be \in \Ftm^*$, we obtain $\ga^t=1$. Thus, we have $\de^t=\ga^t=1$, which contradicts to $\Delta \ne 0$. Therefore, there is no solution when $\Delta \ne 0$.

To sum up, we deduce that $\Nth=(2^n-1)l+1+(2^n-1)((2^m-2)l^2+2l)=(2^m-2)(2^n-1)l^2+3(2^n-1)l+1$.
\end{proof}

The following lemma due to Dobbertin et al. \cite{DFHR} describes the possible number of solutions to certain equation.

\begin{lemma}{\cite[Lemma 22]{DFHR}}\label{lemmaDob}
For $a,b,c \in \Ftn$, the equation
$$
x^{2^r+1}+ax^{2^r}+bx+c=0
$$
has either $0,1,2$ or $2^{r_0}+1$ solutions in $\Ftn$, where $r_0=(r,n)$.
\end{lemma}

For any $z \in S$ and $a,b \in \Fq$ with $a\ba{a}+b\ba{b} \ne 0$, we define the fractional linear transformation (FLT) on $S$ as
$$
\Phi_{a,b}(z)=\frac{az+b}{\ba{b}z+\ba{a}}.
$$
It is straightforward to verify that the FLT is well-defined and induces a permutation on $S$. In particular, the composition of two FLTs is also an FLT. More precisely, we have
$$
\Phi_{a_3,b_3}=\Phi_{a_1,b_1}\Phi_{a_2,b_2},
$$
where
$$
\begin{pmatrix} a_3 & b_3 \\ \ba{b}_3 & \ba{a}_3 \end{pmatrix}
=\begin{pmatrix} a_1 & b_1 \\ \ba{b}_1 & \ba{a}_1 \end{pmatrix}
\begin{pmatrix} a_2 & b_2 \\ \ba{b}_2 & \ba{a}_2 \end{pmatrix}
$$
and
$$
a_3\ba{a}_3+b_3\ba{b}_3=(a_1\ba{a}_1+b_1\ba{b}_1)(a_2\ba{a}_2+b_2\ba{b}_2) \ne 0.
$$

Now we proceed to consider the value distribution of
$$
T_1(a,b)=\sum_{x\in \Ftn} (-1)^{\Tn(ax^{d_1}+bx^{d_2})}.
$$

\begin{theorem}\label{vdbinaryfour}
Let $p=2$, $q=2^n$ and $n=2m$ with $m \ge 2$. Given an integer $1 \le k \le m$, set $s_1=2^{k-1}t-\frac{t-1}{2}$ and $s_2=2^{k-1}t+\frac{t+1}{2}$ with odd integer $1 \le t \le 2^m+1$. Define $d_1=s_1(2^m-1)+1$ and $d_2=s_2(2^m-1)+1$. Assume $(2^k-1)t,(2^k+1)t \not\equiv 0 \pmod{2^m+1}$ and $l=(t,2^m+1)$. If one of the following condition holds:
\begin{enumerate}
\item[i)] $m \equiv -1 \pmod k$,
\item[ii)] $(k,2m)=1$,
\end{enumerate}
then the value distribution of $T_1(a,b)$ is listed in Table \ref{table3}.
\end{theorem}
\begin{proof}
By 2) of Lemma~\ref{Niho}, we have $T_1(a,b)=(V(a,b)-1)2^m$, where $V(a,b)$ is the number of $z \in S$ satisfying
\begin{equation}\label{Nihoeqn1}
\ba{b}z^{(2^k+1)t}+\ba{a}z^{2^{k}t}+az^{t}+b=0.
\end{equation}
If $(a,b)=(0,0)$, it is easy to see that $V(a,b)=2^m+1$ and $T_1(a,b)$ takes the trivial value $2^{2m}$. Our main task is to prove that $T_1(a,b)$ takes at most four nontrivial values if Condition i) or Condition ii) holds.

Setting $w=z^t$, Equation (\ref{Nihoeqn1}) becomes
\begin{equation}\label{Nihoeqn2}
\ba{b}w^{2^k+1}+\ba{a}w^{2^{k}}+aw+b=0.
\end{equation}
Since $l=(t,2^m+1)$, each solution $w \in S_l$ of (\ref{Nihoeqn2}) corresponds to $l$ solutions of (\ref{Nihoeqn1}). Below, we focus on Equation (\ref{Nihoeqn2}) and study the number of its solutions in $S_l$.

At first, assume Condition i) holds. If $a \ne 0$ and $b=0$, we have $\ba{a}w^{2^{k}-1}+a=0$, which implies $w^{2^k-1}=\frac{a}{\ba{a}} \in S$. Noting that $m \equiv -1 \pmod k$, it is straight forward to verify that
$$
(2^k-1,2^m+1)=\begin{cases}
                  1 & \text{if $k$ is odd,} \\
                  3 & \text{if $k$ is even.}
                  \end{cases}
$$
Hence, $(\ref{Nihoeqn2})$ has no more than $3$ solutions in $S_l$. A similar treatment shows that $(\ref{Nihoeqn2})$ has no more than $3$ solutions in $S_l$ when $a=0$ and $b \ne 0$. If $ab \ne 0$, we continue our analysis using a technique proposed in \cite[Proposition 1]{Dob}. Suppose $a\ba{a}+b\ba{b}=0$, we have $(\ba{b}w^{2^k}+a)(w+\frac{b}{a})=0$, which has no more than $2$ solutions in $S_l$. When $a\ba{a}+b\ba{b} \ne 0$, we consider the following FLT:
$$
\Phi_{a,b}(w)=\frac{aw+b}{\ba{b}w+\ba{a}}.
$$
By (\ref{Nihoeqn2}), we have
$$
w^{2^k}=\frac{aw+b}{\ba{b}w+\ba{a}}=\Phi_{a,b}(w).
$$
Since $m \equiv -1 \pmod k$, there exists an integer $i$ such that $ki=m+1$. Applying $\Phi_{a,b}$ on both sides of the above equation with $i-1$ times, we obtain
$$
w^{2^{ki}}=\Phi_{a^{\prime},b^{\prime}}(w),
$$
where
$$
\begin{pmatrix} a^{\prime} & b^{\prime} \\ \ba{b}^{\prime} & \ba{a}^{\prime} \end{pmatrix}
=\begin{pmatrix} a & b \\ \ba{b} & \ba{a}\end{pmatrix}^{i}.
$$
For $w \in S$, we have $w^{2^{ki}}=w^{2^{m+1}}=w^{-2}$. It follows that
$$
a^{\prime}w^3+b^{\prime}w^2+\ba{b}^{\prime}w+\ba{a}^{\prime}=0.
$$
Hence, $(\ref{Nihoeqn2})$ has no more than $3$ solutions in $S_l$. To sum up, when $(a,b) \ne (0,0)$, $(\ref{Nihoeqn2})$ has either $0,1,2$ or $3$ solutions in $S_l$. This implies that $V(a,b) \in \{0,l,2l,3l\}$ when $(a,b) \ne (0,0)$.

Secondly, assume Condition ii) holds. If $b=0$, (\ref{Nihoeqn2}) becomes $\ba{a}w^{2^{k}-1}+a=0$. Since $(k,2m)=1$, it is easy to see that this equation has no more than $1$ solution in $S_l$. If $b \ne 0$, by Lemma~\ref{lemmaDob}, $(\ref{Nihoeqn2})$ has either $0,1,2$ or $3$ solutions in $S_l$. To sum up, when $(a,b) \ne (0,0)$, $(\ref{Nihoeqn2})$ has either $0,1,2$ or $3$ solutions in $S_l$. This implies that $V(a,b) \in \{0,l,2l,3l\}$ when $(a,b) \ne (0,0)$.

Consequently, we have shown that $T_1(a,b)$ takes at most four nontrivial values $\{-2^m,(l-1)2^m,(2l-1)2^m,(3l-1)2^m\}$ if Condition i) or Condition ii) holds. The frequencies of these values easily follow from Lemma~\ref{moment} and Lemma~\ref{N3}. The proof is now complete and we list the value distribution in Table \ref{table3}.
\end{proof}

\begin{remark}
When $k$ is odd, each pair $(k,m)$ meeting the Condition i) always satisfies the Condition ii).
\end{remark}

\begin{table*}
\begin{center}
\ra{1.8}
\caption{Value distribution of Theorem~\ref{vdbinaryfour}}\label{table3}
\begin{tabular}{|c|c|}
\hline
Value  &  Frequency \\ \hline
$2^{2m}$     &  $1$ \\ \hline
$(3l-1)2^{m}$ & $\frac{(2^{2m}-1)(2^m+1-2l)(2^m+1-l)}{6l^3}$ \\ \hline
$(2l-1)2^{m}$ & $\frac{(2^{2m}-1)((2^m+3)l-2^m-1)(2^m+1-l)}{2l^3}$ \\ \hline
$(l-1)2^{m}$ & $\frac{(2^{2m}-1)((2^{2m+1}+2^{m+2}+6)l^2-(2^{2m+1}+7\cdot2^m+5)l+(2^m+1)^2)}{2l^3}$ \\ \hline
$-2^{m}$ & $\frac{(2^{2m}-1)(6(2^{2m}+1)l^3-(6\cdot2^{2m}+9\cdot2^m+11)l^2+(3\cdot2^{2m}+9\cdot2^m+6)l-(2^m+1)^2)}{6l^3}$ \\ \hline
\end{tabular}
\end{center}
\end{table*}

The following theorem is a direct consequence of Theorem~\ref{vdbinaryfour}.

\begin{theorem}\label{thmbinaryfour}
Let $p=2$, $q=2^n$ and $n=2m$ with $m \ge 2$. Given an integer $1 \le k \le m$, set $s_1=2^{k-1}t-\frac{t-1}{2}$ and $s_2=2^{k-1}t+\frac{t+1}{2}$ with odd integer $1 \le t \le 2^m+1$. Define $d_1=s_1(2^m-1)+1$ and $d_2=s_2(2^m-1)+1$. Assume $(2^k-1)t,(2^k+1)t \not\equiv 0 \pmod{2^m+1}$ and $l=(t,2^m+1)$. Suppose one of the following condition holds:
\begin{enumerate}
\item[i)] $m \equiv -1 \pmod k$,
\item[ii)] $(k,2m)=1$.
\end{enumerate}
Then $\code^{\perp}$ is a $[2^n-1,4m,2^{2m-1}-(3l-1)2^{m-1}]$ binary code. Its weight distribution is listed in Table \ref{table4}.
\end{theorem}

\begin{table*}
\begin{center}
\ra{1.8}
\caption{Weight distribution of Theorem~\ref{thmbinaryfour}}\label{table4}
\begin{tabular}{|c|c|}
\hline
Weight  &  Frequency \\ \hline
$0$     &  $1$ \\ \hline
$2^{2m-1}-(3l-1)2^{m-1}$ & $\frac{(2^{2m}-1)(2^m+1-2l)(2^m+1-l)}{6l^3}$ \\ \hline
$2^{2m-1}-(2l-1)2^{m-1}$ & $\frac{(2^{2m}-1)((2^m+3)l-2^m-1)(2^m+1-l)}{2l^3}$ \\ \hline
$2^{2m-1}-(l-1)2^{m-1}$ & $\frac{(2^{2m}-1)((2^{2m+1}+2^{m+2}+6)l^2-(2^{2m+1}+7\cdot2^m+5)l+(2^m+1)^2)}{2l^3}$ \\ \hline
$2^{2m-1}+2^{m-1}$ & $\frac{(2^{2m}-1)(6(2^{2m}+1)l^3-(6\cdot2^{2m}+9\cdot2^m+11)l^2+(3\cdot2^{2m}+9\cdot2^m+6)l-(2^m+1)^2)}{6l^3}$ \\ \hline
\end{tabular}
\end{center}
\end{table*}


Given $m$, the code is determined by two parameters $k$ and $t$. Below, we present some examples concerning the weight distributions of the cyclic codes derived from the above theorem. According to the code table, some of them have the best known parameters.

\begin{example}
When $m=3$, up to the equivalence of $(d_1,d_2)$, a pair $(k,t)$ satisfying the conditions in Theorem~\ref{thmbinaryfour} belongs to $\{(1,1),(1,5),(1,7)\}$. For all these three pairs, $l=(t,2^m+1)=1$. Hence the corresponding three cyclic codes are $[63,12,24]$ binary codes sharing the same weight distribution:
$$
1+588x^{24}+504x^{28}+1827x^{32}+1176x^{36}.
$$
Referring to the code table \cite{Gra}, the best known binary linear code with length $63$ and dimension $12$ has minimum distance $24$. Therefore, our cyclic codes have the best known parameters and are more preferable than linear code in practice.
\end{example}

\begin{example}
When $m=4$, up to the equivalence of $(d_1,d_2)$, a pair $(k,t)$ satisfying the conditions in Theorem~\ref{thmbinaryfour} belongs to $\{(1,1),(1,3),(1,5),(1,7),(1,9),(1,11),(1,13),(1,15)\}$. For all these eight pairs, $l=(t,2^m+1)=1$. Hence the corresponding eight cyclic codes are $[255,16,112]$ binary codes sharing the same weight distribution:
$$
1+10200x^{112}+4080x^{120}+30855x^{128}+20400x^{136}.
$$
Referring to the code table \cite{Gra}, the best known binary linear code with length $255$ and dimension $16$ has minimum distance $112$. Therefore, our cyclic codes have the best known parameters and are more preferable than linear code in practice.
\end{example}

\section{Nonbinary Cyclic Codes With Niho Exponents}\label{sec4}

This section is devoted to the computation of the weight distribution of certain nonbinary cyclic codes with Niho exponents. Accordingly, we focus on the value distribution of $T_2(a,b)$. Throughout this section, we fix $d_1=s_1(p^m-1)+1$ and $d_2=s_2(p^m-1)+1$ where $s_1=\frac{t+2}{4}$ and $s_2=\frac{3t+2}{4}$ for some $t \equiv 2 \pmod 4$. To ensure that $d_1$ and $d_2$ are not equivalent, we have $t \not\equiv 0 \pmod{p^m+1}$. Moreover, set $s_1=\frac{t+2}{4}$, $s_2=\frac{3t+2}{4}$, $s_1^{\prime}=\frac{t^{\prime}+2}{4}$ and $s_2^{\prime}=\frac{3t^{\prime}+2}{4}$. Suppose $s_1+s_1^{\prime} \equiv 1 \pmod {p^m+1}$ and $s_2+s_2^{\prime} \equiv 1 \pmod {p^m+1}$. Then we have $t+t^{\prime} \equiv 0 \pmod {4(p^m+1)}$. Namely, if $t+t^{\prime} \equiv 0 \pmod {4(p^m+1)}$, we obtain two equivalent pairs of Niho exponents. Hence, we can restrict $t$ in the range $1 \le t \le 4(p^m+1)$. Below, we will determine the value distribution of $T_2(a,b)$.

As a preparation, we have the following lemma.

\begin{lemma}\label{N3pary}
Let $p$ be an odd prime and $q=p^n$. If $l=(t,p^m+1)$, then
\begin{enumerate}
\item[1)] $\Nt=\frac{(p^n-1)}{2}l+1$.
\item[2)] $\Nth=\frac{(p^m-2)(p^{n}-1)}{4}l^2+\frac{3(p^{n}-1)}{2}l+1$.
\end{enumerate}
\end{lemma}

The proof of this lemma is somewhat lengthy and we present it in the Appendix. Now we proceed to determine the value distribution of
$$
T_2(a,b)=\sum_{x\in \Fq} \zp^{\Tn(ax^{d_1}+bx^{d_2})},
$$
where $p$ is an odd prime.

\begin{theorem}\label{vdparyfour}
Assume $n=2m$ with $m \ge 1$. Let $p$ be an odd prime and $q=p^n$ be a prime power. Given a positive integer $t$ with $t \equiv 2 \pmod 4$ and $t \not\equiv 0 \pmod{p^m+1}$, set $s_1=\frac{t+2}{4}$ and $s_2=\frac{3t+2}{4}$. Define $d_1=s_1(p^m-1)+1$, $d_2=s_2(p^m-1)+1$ and $l=(t,p^m+1)$. Then the value distribution of $T_2(a,b)$ is listed in Table \ref{table5}.
\end{theorem}
\begin{proof}
By Lemma~\ref{Nihopary}, we have $T_2(a,b)=(W(a,b)-1)p^m$, where $W(a,b)$ is the number of $u \in S$ satisfying
$$
\ba{b}u^{\frac{3t}{2}}+\ba{a}u^{t}+au^{\frac{t}{2}}+b=0.
$$
If $(a,b)=(0,0)$, we have $W(a,b)=p^m+1$ and $T_2(a,b)$ takes the trivial value $p^{2m}$. If $(a,b) \ne (0,0)$, since $(\frac{t}{2},p^m+1)=\frac{l}{2}$, the above equation clearly has either $0, \frac{l}{2}, l$ or $\frac{3l}{2}$ solutions in $S$. Namely, $W(a,b) \in \{0, \frac{l}{2}, l, \frac{3l}{2}\}$. Thus $T_2(a,b)$ takes four nontrivial values $\{-p^m, (\frac{l}{2}-1)p^m, (l-1)p^m, (\frac{3l}{2}-1)p^m\}$ when $(a,b) \ne (0,0)$. The frequencies of these values easily follow from Lemma~\ref{moment} and Lemma~\ref{N3pary}. The proof is now complete and we list the value distribution in Table \ref{table5}.
\end{proof}

\begin{table*}
\begin{center}
\ra{1.8}
\caption{Value distribution of Theorem~\ref{vdparyfour}}\label{table5}
\begin{tabular}{|c|c|}
\hline
Value  &  Frequency \\ \hline
$p^{2m}$     &  $1$ \\ \hline
$(\frac{3l}{2}-1)p^m$ & $\frac{2(p^{2m}-1)(l-p^m-1)(l-2p^m-2)}{3l^3}$ \\ \hline
$(l-1)p^m$ & $\frac{(p^{2m}-1)(2p^m+2-(p^m+3)l)(l-2p^m-2)}{l^3}$ \\ \hline
$(\frac{l}{2}-1)p^m$ & $\frac{2(p^{2m}-1)((p^{2m}+2p^m+3)l^2-(2p^{2m}+7p^m+5)l+2(p^m+1)^2)}{l^3}$ \\ \hline
$-p^m$ & $\frac{(p^{2m}-1)(3(p^{2m}+1)l^3-(6p^{2m}+9p^m+11)l^2+6(p^{2m}+3p^m+2)l-4(p^m+1)^2)}{3l^3}$ \\ \hline
\end{tabular}
\end{center}
\end{table*}

By Lemma~\ref{Nihopary}, $T_2(\la a,\la b)=T_2(a,b)$ for any $\la \in \Fp^*$. Therefore, we can easily deduce the following theorem by Theorem~\ref{vdparyfour}.

\begin{theorem}\label{thmparyfour}
Assume $n=2m$ with $m \ge 1$. Let $p$ be an odd prime and $q=p^n$ be a prime power. Given a positive integer $t$ with $t \equiv 2 \pmod 4$ and $t \not\equiv 0 \pmod{p^m+1}$, set $s_1=\frac{t+2}{4}$ and $s_2=\frac{3t+2}{4}$. Define $d_1=s_1(p^m-1)+1$, $d_2=s_2(p^m-1)+1$ and $l=(t,p^m+1)$. Then $\code^{\perp}$ is a $[p^n-1,4m,(p^m-p^{m-1})(p^m+1-\frac{3l}{2})]$ $p$-ary code. Furthermore, the weight distribution of $\code^{\perp}$ is listed in Table \ref{table6}.
\end{theorem}


\begin{table*}
\begin{center}
\ra{1.8}
\caption{Weight distribution of Theorem~\ref{thmparyfour}}\label{table6}
\begin{tabular}{|c|c|}
\hline
Weight  &  Frequency \\ \hline
$0$     &  $1$ \\ \hline
$(p^m-p^{m-1})(p^m+1-\frac{3l}{2})$ & $\frac{2(p^{2m}-1)(l-p^m-1)(l-2p^m-2)}{3l^3}$ \\ \hline
$(p^m-p^{m-1})(p^m+1-l)$ & $\frac{(p^{2m}-1)(2p^m+2-(p^m+3)l)(l-2p^m-2)}{l^3}$ \\ \hline
$(p^m-p^{m-1})(p^m+1-\frac{l}{2})$ & $\frac{2(p^{2m}-1)((p^{2m}+2p^m+3)l^2-(2p^{2m}+7p^m+5)l+2(p^m+1)^2)}{l^3}$ \\ \hline
$(p^m-p^{m-1})(p^m+1)$ & $\frac{(p^{2m}-1)(3(p^{2m}+1)l^3-(6p^{2m}+9p^m+11)l^2+6(p^{2m}+3p^m+2)l-4(p^m+1)^2)}{3l^3}$ \\ \hline
\end{tabular}
\end{center}
\end{table*}

Given $p$ and $m$, the code is determined by one parameter $t$. Below, we present a few examples concerning the weight distribution of $p$-ary cyclic codes derived from the above theorem.

\begin{example}
Setting $p=3$, $m=3$ and $t=14$, we have $q=729$, $d_1=105$, $d_2=287$ and $l=(t,p^m+1)=14$. The corresponding cyclic code $\code^{\perp}$ is a $[728,12,126]$ ternary code with weight distribution:
$$
1+104x^{126}+4056x^{252}+70304x^{378}+456976x^{504}.
$$
\end{example}

\begin{example}
For $p=5$, $m=2$ and $2 \le t \le 50$ with $t \equiv 2 \pmod 4$ and $t \ne 26$, we can obtain twelve cyclic codes with $q=625$ and $l=(t,p^m+1)=2$. All these cyclic codes are $[624,8,460]$ $5$-ary codes with the same weight distribution:
$$
1+62400x^{460}+15600x^{480}+187824x^{500}+124800x^{520}.
$$
\end{example}

\section{Conclusion}\label{sec5}

In this paper, we consider the weight distributions of cyclic codes with Niho exponents. As is well known, the determination of weight distributions essentially relies on the calculation of some exponential sums. In particular, we completely determine the value distribution of $S(a,b)$ and compute the value distributions of $T_1(a,b)$ and $T_2(a,b)$ in some cases. As a direct consequence, we obtain the weight distributions of some binary and nonbinary cyclic codes. More specifically, we produce two classes of binary three-weight and four-weight cyclic codes and a class of nonbinary four-weight cyclic codes. By presenting several examples, we observe that some of them are optimal linear codes and some of them have the best known parameters.

\section*{Appendix}

Here, we give the proof of Lemma~\ref{N3pary}.

\begin{proof}[Proof of Lemma~\ref{N3pary}]
1) This proof is similar to that of 1) in Lemma~\ref{N3} and we omit it here.\\
2) By definition, $\Nth$ is the number of solutions to the equations
\begin{equation}\label{N3paryeqn1}
\left\{\begin{array}{c}
             x^{d_1}+y^{d_1}+z^{d_1}=0 \\
             x^{d_2}+y^{d_2}+z^{d_2}=0
             \end{array}
\right., \quad x,y,z \in \Fq.
\end{equation}
When $z=0$, there are exactly $\Nt=\frac{(p^n-1)}{2}l+1$ solutions.

When $z\ne 0$, the situation is more involved. By setting $u=-\frac{x}{z}$ and $v=-\frac{y}{z}$, we only need to consider the system
\begin{equation}\label{syspary}
\left\{\begin{array}{c}
             u^{d_1}+v^{d_1}=1 \\
             u^{d_2}+v^{d_2}=1
             \end{array}
\right., \quad u,v \in \Fq.
\end{equation}
Each solution of (\ref{syspary}) corresponds to $p^n-1$ solutions of (\ref{N3paryeqn1}). If $u=0$ or $v=0$, it is easy to see that the system (\ref{syspary}) has exactly $\frac{l}{2}$ solutions.
If $uv \ne 0$, we split our discussion into the following four cases:
\begin{enumerate}
\item[i)] $u \in Q$ and $v \in Q$,
\item[ii)] $u \in Q$ and $v \in NQ$,
\item[iii)] $u \in NQ$ and $v \in Q$,
\item[iv)] $u \in NQ$ and $v \in NQ$.
\end{enumerate}

Recall that each $x \in Q$ (resp. $x \in NQ$) can be expressed twice as $x=yz$ and $x=(-y)(-z)$ (resp. $x=\theta yz$ and $x=\theta (-y)(-z)$) when $y$ ranges over $\Fpm^*$ and $z$ ranges over $S$. Moreover, $\Fpm^* \cap S =\{\pm1\}$. We will deal with these four cases respectively.

For Case i), we can write $u=\al\de$ and $v=\be\ga$, where $\al,\be \in \Fpm^*$ and $\de,\ga \in S$. Thus the system (\ref{syspary}) can be rewritten as
\begin{equation}\label{syspary2}
\left\{\begin{array}{c}
  \al\de^{-\frac{t}{2}}+\be\ga^{-\frac{t}{2}}=1 \\
  \al\de^{-\frac{3t}{2}}+\be\ga^{-\frac{3t}{2}}=1
\end{array}
\right..
\end{equation}
Note that
$$
\Delta=
\begin{vmatrix}
  \de^{-\frac{t}{2}} & \ga^{-\frac{t}{2}}\\
  \de^{-\frac{3t}{2}} & \ga^{-\frac{3t}{2}}
\end{vmatrix}
=\de^{-\frac{t}{2}}\ga^{-\frac{3t}{2}}-\de^{-\frac{3t}{2}}\ga^{-\frac{t}{2}}.
$$

If $\Delta=0$, we have $\de^t=\ga^t$, i.e., $\ga^{\frac{t}{2}}=\pm\de^{\frac{t}{2}}$. When $\ga^{\frac{t}{2}}=\de^{\frac{t}{2}}$, comparing with (\ref{syspary2}), we have $\al+\be=\de^{\frac{t}{2}}$. Noting that $\Fpm^* \cap S =\{\pm1\}$, we have $\al+\be=\de^{\frac{t}{2}}=\pm1$. There are $\frac{l^2}{4}$ pairs of $(\de,\ga)$ such that $\de^{\frac{t}{2}}=\ga^{\frac{t}{2}}=1$ or $\de^{\frac{t}{2}}=\ga^{\frac{t}{2}}=-1$. Moreover, for each pair $(\de,\ga)$, there are $p^m-2$ pairs of $(\al,\be)$, such that $\al+\be=1$ and $\al\be \ne 0$. Hence, there are $\frac{(p^m-2)}{2}l^2$ tuples of $(\al,\be,\de,\ga)$ satisfying (\ref{syspary2}) when $\ga^{\frac{t}{2}}=\de^{\frac{t}{2}}$. A similar treatment shows there are $\frac{(p^m-2)}{2}l^2$ tuples of $(\al,\be,\de,\ga)$ satisfying (\ref{syspary2}) when $\ga^{\frac{t}{2}}=-\de^{\frac{t}{2}}$. Since both $u$ and $v$ have been expressed twice, there are $\frac14(\frac{(p^m-2)}{2}l^2+\frac{(p^m-2)}{2}l^2)=\frac{(p^m-2)}{4}l^2$ solutions of (\ref{syspary}) when $\Delta=0$.

If $\Delta \ne 0$, i.e., $\de^t \ne \ga^t$, solving the system (\ref{sysbinary2}) yields
\begin{align*}
\al&=\frac{1-\ga^{t}}{\de^{-\frac{t}{2}}(1-\de^{-t}\ga^{t})},\\ \be&=\frac{1-\de^{t}}{\ga^{-\frac{t}{2}}(1-\de^{t}\ga^{-t})}.
\end{align*}
We are going to show that no solution exists. Since $\al,\be \in \Fpm^*$, by $\al=\ba{\al}$ and $\be=\ba{\be}$, we have
$\de^{2t}=1$ and $\ga^{2t}=1$. Since $\de^t \ne \ga^t$,  we have either $\de^t=1$, $\ga^t=-1$ or $\de^t=-1$, $\ga^t=1$. However, $\de^t=1$ implies $\be=0$ and $\ga^t=1$ implies $\al=0$. Hence, there exists no solution when $\Delta \ne 0$. Totally, there are $\frac{(p^m-2)}{4}l^2$ solutions of (\ref{syspary}) in Case i).

For Case ii), we can write $u=\al\de$ and $v=\theta\be\ga$, where $\al,\be \in \Fpm^*$ and $\de,\ga \in S$. Thus the system (\ref{syspary}) can be rewritten as
\begin{equation}\label{syspary3}
\left\{\begin{array}{c}
  \al\de^{-\frac{t}{2}}+\theta^{d_1}\be\ga^{-\frac{t}{2}}=1 \\
  \al\de^{-\frac{3t}{2}}+\theta^{d_2}\be\ga^{-\frac{3t}{2}}=1
\end{array}
\right..
\end{equation}
Note that
$$
\Delta=
\begin{vmatrix}
  \de^{-\frac{t}{2}} & \theta^{d_1}\ga^{-\frac{t}{2}}\\
  \de^{-\frac{3t}{2}} & \theta^{d_2}\ga^{-\frac{3t}{2}}
\end{vmatrix}
=\theta^{d_2}\de^{-\frac{t}{2}}\ga^{-\frac{3t}{2}}-\theta^{d_1}\de^{-\frac{3t}{2}}\ga^{-\frac{t}{2}}.
$$

If $\Delta=0$, we deduce $\de^t\theta^{\frac{t}{2}(p^m-1)}=\ga^t$. Set $\eta=\theta^{p^m-1}$, then $\eta$ is a generator of $S$. We have $\eta^{\frac{t}{2}}=(\frac{\ga}{\de})^t=\eta^{jt}$ for some integer $j$. This is equivalent to $jt \equiv \frac{t}{2} \pmod {p^m+1}$, which is impossible since $t \equiv 2 \pmod 4$.

If $\Delta \ne 0$, solving the system (\ref{syspary3}) yields
\begin{align*}
\al &= \frac{\de^{\frac{t}{2}}(1-\theta^{-\frac{t}{2}(p^m-1)}\ga^t)}{1-\theta^{-\frac{t}{2}(p^m-1)}\de^{-t}\ga^t},\\
\be &= \frac{\ga^{\frac{t}{2}}(1-\de^t)}{\theta^{d_1}(1-\theta^{\frac{t}{2}(p^m-1)}\de^{t}\ga^{-t})}.
\end{align*}
With $\al=\ba{\al}$ and $\be=\ba{\be}$, we can deduce that $\de^{t}\ga^t=1$ and $\ga^{2t}=\theta^{t(p^m-1)}$. Thus, we have $\al=\frac{\de^{\frac{t}{2}}(1-\theta^{-\frac{t}{2}(p^m-1)}\de^{-t})}{1-\theta^{-\frac{t}{2}(p^m-1)}\de^{-2t}}$. By $\al=\ba{\al}$, we have $\de^t= \pm1$. Thus, $\ga^t=\pm1$. However, this contradicts to $\ga^{2t}=\theta^{t(p^m-1)}$ since $t \not\equiv 0 \pmod {p^m+1}$. Hence, (\ref{syspary}) has no solution in Case ii).

For Case iii), the situation is similar to Case ii) and (\ref{syspary}) has no solution in Case iii).

For Case iv),  we can write $u=\theta\al\de$ and $v=\theta\be\ga$, where $\al,\be \in \Fpm^*$ and $\de,\ga \in S$. Thus the system (\ref{syspary}) can be rewritten as
\begin{equation*}
\left\{\begin{array}{c}
  \al\de^{-\frac{t}{2}}+\be\ga^{-\frac{t}{2}}=\theta^{-d_1} \\
  \al\de^{-\frac{3t}{2}}+\be\ga^{-\frac{3t}{2}}=\theta^{-d_2}
\end{array}
\right..
\end{equation*}
Note that
$$
\Delta=
\begin{vmatrix}
  \de^{-\frac{t}{2}} & \ga^{-\frac{t}{2}}\\
  \de^{-\frac{3t}{2}} & \ga^{-\frac{3t}{2}}
\end{vmatrix}
=\de^{-\frac{t}{2}}\ga^{-\frac{3t}{2}}-\de^{-\frac{3t}{2}}\ga^{-\frac{t}{2}}.
$$

If $\Delta=0$, a similar argument as Case ii) shows that no solution exists. If $\Delta \ne 0$, a similar treatment as Case i) shows that no solution exists. Hence, (\ref{syspary}) has no solution in Case iv).

Combining the four cases discussed above, we can deduce that $\Nth=\frac{(p^n-2)}{2}l+1+(p^n-1)(\frac{(p^m-2)}{4}l^2+l)=\frac{(p^m-2)(p^{n}-1)}{4}l^2+\frac{3(p^{n}-1)}{2}l+1$.
\end{proof}


\begin{thebibliography}{10}

\bibitem{C}
P.~Charpin, ``Cyclic codes with few weights and {N}iho exponents,'' \emph{J.
  Combin. Theory Ser. A}, vol. 108, no.~2, pp. 247--259, 2004.

\bibitem{Del}
P.~Delsarte, ``On subfield subcodes of modified {R}eed-{S}olomon codes,''
  \emph{IEEE Trans. Inform. Theory}, vol.~21, no.~5, pp. 575--576, 1975.

\bibitem{DGZ}
C.~Ding, Y.~Gao, and Z.~Zhou, ``Five families of three-weight ternary cyclic
  codes and their duals,'' \emph{IEEE Trans. Inform. Theory}, vol.~59, no.~12,
  pp. 7940--7946, 2013.

\bibitem{DLMZ}
C.~Ding, Y.~Liu, C.~Ma, and L.~Zeng, ``The weight distributions of the duals of
  cyclic codes with two zeros,'' \emph{IEEE Trans. Inform. Theory}, vol.~57,
  no.~12, pp. 8000--8006, 2011.

\bibitem{DY}
C.~Ding and J.~Yang, ``Hamming weights in irreducible cyclic codes,''
  \emph{Discr. Math.}, vol. 313, no.~4, pp. 434--446, 2013.

\bibitem{DYT}
C.~Ding, Y.~Yang, and X.~Tang, ``Optimal sets of frequency hopping sequences
  from linear cyclic codes,'' \emph{IEEE Trans. Inform. Theory}, vol.~56,
  no.~7, pp. 3605--3612, 2010.

\bibitem{Dob}
H.~Dobbertin, ``One-to-one highly nonlinear power functions on {${\rm
  GF}(2^n)$},'' \emph{Appl. Algebra Engrg. Comm. Comput.}, vol.~9, no.~2, pp.
  139--152, 1998.

\bibitem{DFHR}
H.~Dobbertin, P.~Felke, T.~Helleseth, and P.~Rosendahl, ``Niho type
  cross-correlation functions via {D}ickson polynomials and {K}loosterman
  sums,'' \emph{IEEE Trans. Inform. Theory}, vol.~52, no.~2, pp. 613--627,
  2006.

\bibitem{FLX}
T.~Feng, K.~Leung, and Q.~Xiang, ``Binary cyclic codes with two primitive
  nonzeros,'' \emph{Sci. China Math.}, vol.~56, no.~7, pp. 1403--1412, 2013.

\bibitem{FM}
T.~Feng and K.~Momihara, ``Evaluation of the weight distribution of a class of
  cyclic codes based on index 2 gauss sums,'' \emph{IEEE Trans. Inform.
  Theory}, vol.~59, no.~9, pp. 5980--5984, 2013.

\bibitem{Gra}
M.~Grassl, ``{Bounds on the minimum distance of linear codes and quantum
  codes},'' Online available at {http://www.codetables.de}, 2007, accessed on
  2013-11-04.

\bibitem{HX}
H.~D.~L. Hollmann and Q.~Xiang, ``On binary cyclic codes with few weights,'' in
  \emph{Finite fields and applications ({A}ugsburg, 1999)}.\hskip 1em plus
  0.5em minus 0.4em\relax Berlin: Springer, 2001, pp. 251--275.

\bibitem{LZH}
C.~Li, X.~Zeng, and L.~Hu, ``A class of binary cyclic codes with five
  weights,'' \emph{Sci. China Math.}, vol.~53, no.~12, pp. 3279--3286, 2010.

\bibitem{LF1}
J.~Luo and K.~Feng, ``Cyclic codes and sequences from generalized
  {C}oulter-{M}atthews function,'' \emph{IEEE Trans. Inform. Theory}, vol.~54,
  no.~12, pp. 5345--5353, 2008.

\bibitem{LF2}
------, ``On the weight distributions of two classes of cyclic codes,''
  \emph{IEEE Trans. Inform. Theory}, vol.~54, no.~12, pp. 5332--5344, 2008.

\bibitem{LTW}
J.~Luo, Y.~Tang, and H.~Wang, ``Cyclic codes and sequences: the generalized
  {K}asami case,'' \emph{IEEE Trans. Inform. Theory}, vol.~56, no.~5, pp.
  2130--2142, 2010.

\bibitem{MZLFD}
C.~Ma, L.~Zeng, Y.~Liu, D.~Feng, and C.~Ding, ``The weight enumerator of a
  class of cyclic codes,'' \emph{IEEE Trans. Inform. Theory}, vol.~57, no.~1,
  pp. 397--402, 2011.

\bibitem{Mc}
R.~J. McEliece, ``Irreducible cyclic codes and {G}auss sums,'' in
  \emph{Combinatorics ({P}roc. {NATO} {A}dvanced {S}tudy {I}nst., {B}reukelen,
  1974), {P}art 1: {T}heory of designs, finite geometry and coding
  theory}.\hskip 1em plus 0.5em minus 0.4em\relax Amsterdam: Math. Centrum,
  1974, pp. 179--196. Math. Centre Tracts, No. 55.

\bibitem{M}
M.~Moisio, ``Explicit evaluation of some exponential sums,'' \emph{Finite
  Fields Appl.}, vol.~15, no.~6, pp. 644--651, 2009.

\bibitem{Niho}
Y.~Niho, ``Multivalued cross-correlation functions between two maximal linear
  recursive sequence,'' Ph.D. dissertation, Univ. Southern Calif., Los Angeles,
  1970.

\bibitem{P}
V.~Pless, ``Power moment identities on weight distributions in error correcting
  codes,'' \emph{Inf. Control}, vol.~6, pp. 147--152, 1963.

\bibitem{TM}
A.~Thangaraj and S.~McLaughlin, ``Quantum codes from cyclic codes over {${\rm
  GF}(4^m)$},'' \emph{IEEE Trans. Inform. Theory}, vol.~47, no.~3, pp.
  1176--1178, 2001.

\bibitem{Ve}
G.~Vega, ``The weight distribution of an extended class of reducible cyclic
  codes,'' \emph{IEEE Trans. Inform. Theory}, vol.~58, no.~7, pp. 4862--4869,
  2012.

\bibitem{VM}
G.~Vega and L.~B. Morales, ``A general description for the weight distribution
  of some reducible cyclic codes,'' \emph{IEEE Trans. Inform. Theory}, vol.~59,
  no.~9, pp. 5994--6001, 2013.

\bibitem{WTQYX}
B.~Wang, C.~Tang, Y.~Qi, Y.~Yang, and M.~Xu, ``The weight distributions of
  cyclic codes and elliptic curves,'' \emph{IEEE Trans. Inform. Theory},
  vol.~58, no.~12, pp. 7253--7259, 2012.

\bibitem{Wo}
J.~Wolfmann, ``Weight distributions of some binary primitive cyclic codes,''
  \emph{IEEE Trans. Inform. Theory}, vol.~40, no.~6, pp. 2068--2071, 1994.

\bibitem{X}
M.~Xiong, ``The weight distributions of a class of cyclic codes,'' \emph{Finite
  Fields Appl.}, vol.~18, no.~5, pp. 933--945, 2012.

\bibitem{ZHJYC}
X.~Zeng, L.~Hu, W.~Jiang, Q.~Yue, and X.~Cao, ``The weight distribution of a
  class of {$p$}-ary cyclic codes,'' \emph{Finite Fields Appl.}, vol.~16,
  no.~1, pp. 56--73, 2010.

\bibitem{ZLH}
X.~Zeng, N.~Li, and L.~Hu, ``A class of nonbinary codes and sequence
  families,'' in \emph{Sequences and their applications---{SETA} 2008}, ser.
  Lecture Notes in Comput. Sci.\hskip 1em plus 0.5em minus 0.4em\relax Berlin:
  Springer, 2008, vol. 5203, pp. 81--94.

\bibitem{ZDLZ}
Z.~Zhou, C.~Ding, J.~Luo, and A.~Zhang, ``A family of five-weight cyclic codes
  and their weight enumerators,'' \emph{IEEE Trans. Inform. Theory}, vol.~59,
  no.~10, pp. 6674--6682, 2013.

\end{thebibliography}
\end{document}